\documentclass[journal]{IEEEtran}

\usepackage{amsmath, amssymb, amsthm}
\usepackage{graphicx}
\usepackage{setspace}  % Ensuring spacing support
\usepackage{algorithm, algorithmicx, algpseudocode}
\usepackage{cite}
\usepackage{lipsum} 
\usepackage[colorlinks=true, linkcolor=black, citecolor=black, urlcolor=black]{hyperref}
\usepackage{bm,color,soul}
\usepackage{caption}
\usepackage{subcaption}
\usepackage{balance}
\usepackage{xspace,flushend}

\theoremstyle{plain}
\newtheorem{theorem}{Theorem}
\newtheorem{lemma}{Lemma}

\theoremstyle{definition}

\newcommand{\bmA}{\mathbf A}

\newcommand{\bmh}{\mathbf h}

\newcommand{\bmk}{\mathbf k}

\newcommand{\bmv}{\mathbf v}

\newcommand{\bmW}{\mathbf W}
\newcommand{\bmw}{\mathbf w}

\newcommand{\bmr}{\mathbf r}

\begin{document}

\title{Holographic Joint Communications and Sensing With Cramér-Rao Bounds}  
\author{Chandan~Kumar~Sheemar, Wali Ullah Khan, George Alexandropoulos,~\IEEEmembership{Senior~Member,~IEEE}, \\Jorge Querol,~\IEEEmembership{Member,~IEEE}, and Symeon Chatzinotas,~\IEEEmembership{Fellow,~IEEE}  
 \thanks{C. K. Sheemar, W. U. Khan, J. Querol, and S. Chatzinotas are with the SnT department at the University of Luxembourg (emails:\{chandankumar.sheemar, wali.ullah.khan,  jorge.querol, symeon.chatzinotas\}@uni.lu).} 
 \thanks{G. C. Alexandropoulos is with the Department of Informatics and Telecommunications, National and Kapodistrian University of Athens, 16122 Athens, Greece and with the Department of Electrical and Computer Engineering, University of Illinois Chicago, IL 60601, USA (email: alexandg@di.uoa.gr).} } 
 \maketitle

\begin{abstract}
 Joint Communication and Sensing (JCAS) technology facilitates the seamless integration of communication and sensing functionalities within a unified framework, enhancing spectral efficiency, reducing hardware complexity, and enabling simultaneous data transmission and environmental perception. This paper explores the potential of holographic JCAS systems by leveraging reconfigurable holographic surfaces (RHS) to achieve high-resolution hybrid holographic beamforming while simultaneously sensing the environment. As the holographic transceivers are governed by arbitrary antenna spacing, we first derive exact Cramér-Rao Bounds (CRBs) for azimuth and elevation angles to rigorously characterize the three-dimensional (3D) sensing accuracy. To optimize the system performance, we propose a novel weighted multi-objective problem formulation that aims to simultaneously maximize the communication rate and minimize the CRBs. However, this formulation is highly non-convex due to the inverse dependence of the CRB on the optimization variables, making the solution extremely challenging. To address this, we propose a novel algorithmic framework based on the Majorization-Maximization (MM) principle, employing alternating optimization to efficiently solve the problem. The proposed method relies on the closed-form surrogate functions that majorize the original objective derived herein, enabling tractable optimization. Simulation results are presented to validate the effectiveness of the proposed framework under diverse system configurations, demonstrating its potential for next-generation holographic JCAS systems.
\end{abstract}

\IEEEpeerreviewmaketitle

\section{Introduction} \label{Intro} 
\IEEEPARstart{T}{he} rapid evolution of wireless technologies is driving the demand for networks that go beyond conventional communication functionalities \cite{dang2020should}. As we move towards 6G and beyond, future wireless systems are expected to provide not only high-speed, ultra-reliable connectivity but also environmental awareness and intelligent sensing capabilities \cite{wang2023road}. This shift is fueled by the growing need for autonomous systems, smart cities, industrial automation, and extended reality (XR), all of which require seamless interaction between the digital and physical worlds. Furthermore, as wireless networks become more densely deployed and spectrum resources become increasingly scarce, there is a critical need for technologies that can maximize spectrum efficiency while simultaneously enabling new functionalities \cite{hong20226g}. 

Joint Communication and Sensing (JCAS) represents a groundbreaking paradigm that seamlessly integrates sensing and communication into a unified framework \cite{liu2022integrated,zheng2019radar}. By enabling real-time situational awareness and adaptive decision-making, JCAS transcends the limitations of traditional networks, which treat sensing and communication as distinct functions \cite{sheemar2023full}. Instead, it harnesses shared spectrum, hardware, and signal processing to perform both tasks concurrently, unlocking unparalleled efficiency in spectrum utilization, energy consumption, hardware simplicity, and network intelligence. This dual-functionality is particularly transformative for applications demanding real-time environmental perception, effectively turning wireless networks into dynamic, perceptive, and interactive systems. By bridging the digital and physical realms, JCAS promises unprecedented opportunities across diverse industries, fundamentally redefining how networks perceive, interpret, and respond to the world in real-time \cite{wymeersch2021integration}.

The evolution of JCAS is being paralleled by a transformative shift in antenna technologies toward holographic arrays \cite{huang2020holographic}, which are poised to revolutionize how wireless networks transmit and receive information. Unlike traditional antenna systems \cite{sheemar2022hybrid}, holographic transceivers leverage densely deployed sub-wavelength antennas, enabling real-time manipulation of electromagnetic (EM) waves and unprecedented control over wireless channels \cite{an2023tutorial}. In contrast to conventional systems that depend on bulky, high-cost, and power-intensive phased arrays, next-generation holographic systems can be realized using low-cost, lightweight, and energy-efficient reconfigurable holographic surfaces (RHS)  \cite{gong2023holographic}.. By integrating software-reconfigurable structures and high-resolution wavefront shaping, these systems can dynamically shape and steer wavefronts with extreme precision \cite{you2022energy}, achieving superior spectral efficiency, enhanced spatial resolution, and dynamic reconfigurability. This advancement paves the way for intelligent, software-controlled wireless environments, marking a significant leap forward in wireless technology.

 The potential of RHS to enable holographic JCAS is particularly promising \cite{iacovelli2024holographic}. Unlike traditional architectures, RHS-assisted transceivers leverage holographic beamforming to simultaneously facilitate communication and sensing by providing fine-grained control over the radiation pattern and real-time reconfiguration of the EM wavefront \cite{zhang2022holographic}. This capability allows JCAS systems to dynamically optimize signal transmission and sensing functionalities, achieving an unprecedented level of flexibility and adaptability. Despite its significant potential, holographic JCAS is still in its early stages of development, necessitating further advancements in algorithm design, performance analysis, and system optimization to comprehensively evaluate its benefits and practical feasibility.

\subsection{State-of-the-Art}

In the subsequent sections, we present a comprehensive review of the literature, highlighting pioneering contributions in the field of JCAS. 

The work \cite{liu2021cramer} proposed the design of beamforming for JCAS to jointly optimize target estimation through Cramer-Rao bound (CRB) minimization while maintaining desirable communication quality with minimum signal-to-interference-plus-noise (SINR) constraints. The study \cite{qin2024cramer} explored a movable antenna-assisted multi-user JCAS system, where optimization to minimize CRB is formulated as NP-hard due to coupled variables. An alternating optimization framework using semidefinite relaxation (SDR) and successive convex approximation (SCA) is proposed for an efficient solution. The paper \cite{9652071} designed multiple-input multiple-output (MIMO) beamforming for JCAS that supports both point target and extended target estimation while ensuring the required SINR for communication users. A closed-form solution is derived for the single-user case, while the multi-user case is addressed via SDR, ensuring global optimality. In \cite{song2023intelligent}, the authors proposed reconfigurable intelligent surface (RIS)-assisted non line-of-sight (NLoS) wireless sensing. The joint design of base station (BS) transmit beamforming and RIS reflective beamforming is achieved using SCA, SDR, and alternating optimization to minimize the CRB. The work \cite{zhao2024joint} proposed a robust joint beamforming framework for JCAS to enable multi-target sensing. They ensure reliable target estimation by formulating the problem as a min-max optimization that minimizes the maximum CRB for the direction of arrival (DOA) estimation while maintaining SINR constraints for communication users. 

Studies on JCAS aided with reconfigurable intelligent surfaces (RIS) are presented in the following. The study \cite{wang2021joint} employed RIS to reduce multi-user interference in JCAS systems under sensing constraints. It jointly optimized the constant-modulus waveforms and discrete RIS phase shifts based on alternating optimization. The study \cite{10364760} utilized RIS to enhance degrees of freedom (DoF) in JCAS optimization, proposing two beamforming techniques: one maximizing radar mutual information under communication constraints and another using a Riemannian gradient approach for weighted mutual information optimization. In \cite{10042425}, the authors proposed a joint active and passive beamforming design for RIS-assisted JCAS, considering target size. They derive a closed-form detection probability based on target illumination and introduce ultimate detection resolution. The work \cite{10746496} proposed JCAS in a full-duplex cell-free (FD-CF) MIMO system with RIS assistance, using multiple access points for target detection and uplink communication. In \cite{sheemar2023full_mag}, a novel optimization design for FD JCAS to jointly handle the self-interference with RIS is proposed. In \cite{10274660} a joint optimization framework for beamforming, power allocation, and signal processing in an RIS-assisted FD uplink system is proposed. The work \cite{10845212} investigated a RIS-enhanced cognitive JCAS system for 6G networks, showing that increasing SNR reduces the position error bound of mobile sensors.

Literature on holographic JCAS is available in \cite{zhang2022holographic,zhang2023holographic, GIS2023, GA2024,adhikary2023integrated,adhikary2024holographic,zhao2025performance,liu2024holographic}. In \cite{zhang2022holographic},  the authors proposed an RHS-assisted JCAS system using holographic beamforming to replace traditional MIMO arrays, enhancing radar performance while meeting communication SINR requirements. They develop an iterative algorithm to optimize digital and analog beamformers and provide a theoretical lower bound for maximum beampattern gain. For such work, a prototype was also presented in \cite{zhang2023holographic}. In \cite{GIS2023}, the authors studied the problem of jointly maximizing the SINR for communication users and radar for an FD holographic JCAS system operating in the near-field. In \cite{GA2024}, the authors presented simultaneous target tracking while enabling
multi-user communication for a holographic system. The work \cite{adhikary2023integrated} proposed an artificial intelligence (AI) framework for an RHS-assisted holographic JCAS and localization.
In \cite{adhikary2024holographic}, the author extends their AI framework to the case of cell-free JCAS. In \cite{zhao2025performance}, the author derived the closed-form expression for outage probabilities for the JCAS under perfect and imperfect channel state information (CSI). The study \cite{liu2024holographic} presented a novel beamforming design for holographic JCAS for continuous aperture holographic surfaces. An alternating optimization algorithm is developed to solve the problem, with transmit beamforming addressed through feasibility-checking sub-problems and receive beamforming optimized using a Rayleigh quotient-based method.  

\subsection{Motivation}

The field of JCAS has witnessed remarkable progress in recent years, with significant contributions also for holographic JCAS \cite{zhang2022holographic,zhang2023holographic,GIS2023, GA2024,adhikary2023integrated,adhikary2024holographic,zhao2025performance,liu2024holographic}.  
While substantial efforts have been dedicated to optimizing beamforming designs and enhancing hardware implementations, a fundamental limitation remains: the CRB, a cornerstone of estimation theory, has yet to be explicitly integrated into the design of holographic JCAS systems. This omission is particularly concerning given the unique capabilities of holographic surfaces, which can enable advanced three-dimensional (3D) sensing by leveraging both elevation and azimuth angles capabilities that far exceed those of traditional uniform linear arrays (ULAs). Existing CRB derivations in the current literature for JCAS \cite{liu2021cramer,qin2024cramer,wang2021joint,song2023intelligent}, which are tailored for ULAs, are inadequate for capturing the complexities of the holographic JCAS systems, including complex geometry with sub-wavelength antenna spacing.  

Furthermore, note that in the field of JCAS, a common practice has been to address optimization problems by imposing constraints on one functionality to optimize the other \cite{liu2021cramer,qin2024cramer,zhang2022holographic}. This approach often involves prioritizing either communication or sensing while treating the other as a secondary objective, subject to predefined performance limits. While this method simplifies the problem by partially decoupling the two functionalities, it inherently limits the system's overall potential. By artificially constraining one objective to favour the other, such designs fail to explore the full spectrum of possible solutions that could achieve superior performance in both domains simultaneously. While some works have considered multi-objective optimization frameworks, such works are limited to joint mutual-information maximization \cite{meng2024multi} or SINR maximization \cite{GIS2023}. Such an approach is desirable as due to similarities, it allows the use of traditional optimization tools to maximize performance. However, in practice, communications and sensing can have fundamentally different requirements, leading to distinct and sometimes conflicting objective functions that cannot be adequately captured by a single metric like mutual information or SINR.

This limitation highlights the critical need for general multi-objective optimization frameworks that can simultaneously address the diverse and often competing goals of communication and sensing. By adopting general multi-objective frameworks, we can explore the Pareto-optimal solutions that represent the best possible trade-off between communication and sensing performance, enabling the design of systems that are much more versatile and efficient.

\subsection{Main Contributions}
In this work, we investigate an RHS-assisted JCAS system designed to serve a communication user while simultaneously detecting a target located in the 3D sensing area. The proposed system leverages a BS equipped with a hybrid beamforming architecture, combining a low-dimensional digital beamforming providing the flexibility for adaptive signal processing and interference management, and the high-dimensional analog holographic beamforming implemented via RHS, which enables dynamic control over electromagnetic wavefronts with high precision and flexibility. Firstly, we derive the CRBs for both the elevation and azimuth angles for the holographic JCAS systems.
By deriving such bounds, we provide a theoretical foundation for characterizing the fundamental limits of estimation accuracy in 3D holographic sensing, which is linked to the different parameters of the systems such as sub-wavelength antenna spacing and hybrid holographic beamformers. Subsequently, we propose a novel weighted multi-objective optimization framework which simultaneously considers the maximization of communication rate and the minimization of the CRBs for elevation and azimuth angles as objectives. However, this problem is highly complex due to the direct relationship between the communication rate and the optimization variables, which inversely affect the CRBs.  

To tackle this complex problem, we propose a novel algorithmic framework based on alternating optimization and the Majorization-Maximization (MM) method. This framework is designed to address the intricate trade-offs between communication and sensing objectives by decomposing the problem into manageable subproblems. At its core, the MM approach enables the construction of surrogate functions that approximate the challenging inverse terms present in the CRB expressions with simplified expression, thus enabling tractable optimization. Based on the simplified structure, the digital and the analog holographic beamforming are alternatively updated based on the dominant eigenvector solution and projected gradient ascent (PGA) optimization, respectively.
Simulation results are presented to validate the effectiveness and robustness of the proposed framework. These results demonstrate significant improvements in both communication and sensing performance compared to the benchmark scheme. In summary, the main contributions of our work are the following:
\begin{itemize}
\item We investigate an RHS-assisted holographic JCAS system with a hybrid beamforming architecture, combining digital beamforming with analog holographic beamforming implemented with an RHS.
\item We derive the exact CRBs for both elevation and azimuth angles, providing a theoretical foundation for characterizing the fundamental limits of 3D sensing accuracy in holographic JCAS systems.

\item We propose a novel weighted multi-objective optimization framework to jointly maximize the communication rate and minimize the CRBs for the elevation and azimuth angles.

\item We develop a novel efficient algorithmic framework based on MM, which simplifies the problem structure by constructing surrogate functions and enables tractable optimization through iterative updates.

\item We validate the proposed framework through extensive simulations, demonstrating significant improvements in both communication and sensing performance.
\end{itemize}

\emph{Paper Organization:}  The rest of the paper is structured as follows. Section \ref{section_2} presents the system model, derives the CRBs and formulates the multi-objective optimization problem for holographic JCAS. Section \ref{section_3} develops the MM-based framework and presents a novel algorithmic design to jointly solve the multi-objective problem. Finally, Section \ref{section_4} presents the simulation results, and Section \ref{section_5} concludes the paper.

\emph{Notations:} In this paper, we adopt a consistent set of notations: Scalars are denoted by lowercase or uppercase letters, while vectors and matrices are represented by bold lowercase and bold uppercase letters, respectively. The transpose, Hermitian transpose, and inverse of a matrix $\mathbf{X}$ are denoted by $\mathbf{X}^\mathrm{T}$, $\mathbf{X}^\mathrm{H}$, and $\mathbf{X}^{-1}$, respectively. %Sets are indicated by calligraphic letters (e.g., $\mathcal{X}$), and their cardinality is represented by $|\mathcal{X}|$. Finally, $|\cdot|$ denotes the $l_2$-norm.

\section{System Model and Problem Formulation} \label{section_2}

\begin{figure}[!t]
\centering
\includegraphics [width=0.49\textwidth]{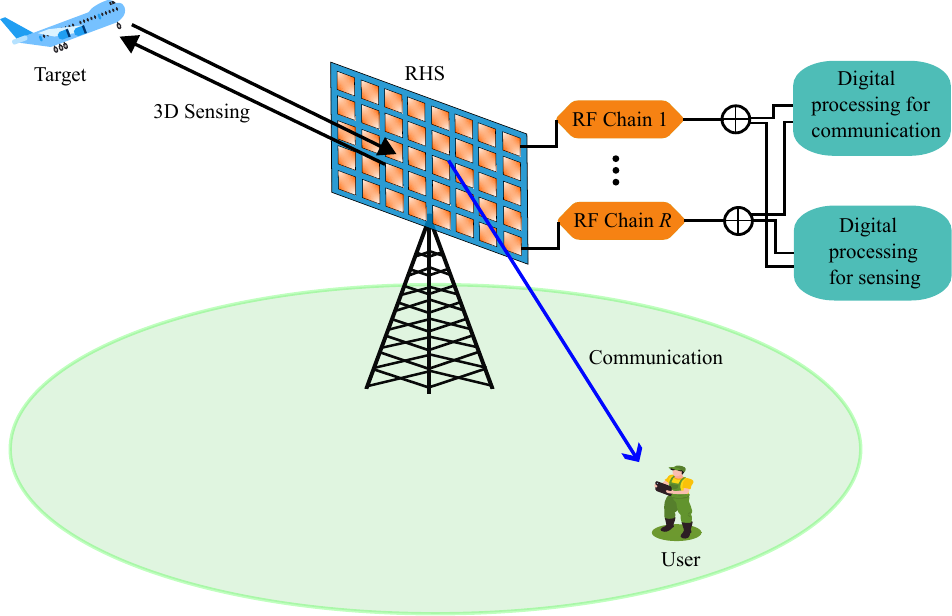}
\caption{An RHS-assisted holographic JCAS with hybrid transceiver.}
\label{SM}
\end{figure}
\subsection{Scenario Description} 
We consider a holographic JCAS system, where an RHS-assisted BS transmits signal to a communication user by employing hybrid beamforming architecture, as shown in Fig. \ref{SM}. The BS is assumed to have $R$ RF chains and the digital beamformer is denoted as $\mathbf{v_d} \in \mathbb{C}^{R \times 1}$. Each RF chain feeds a corresponding element of the high-dimensional holographic beamformer, which transforms the high-frequency electrical signal into a propagating electromagnetic wave. The number of RF chains $R$ is assumed to be equal to the number of feeds $K$ at the RHS. The holographic beamformer denoted as $\bmW \in \mathbb{C}^{M \times K}$ is responsible for generating highly directional beams at the antenna front-end. We consider a square RHS consisting of $M$ elements arranged in a two-dimensional grid with $\sqrt{M}$ elements along the $x$-axis and $y$-axis. Each element of the holographic matrix $\bmW$ is of the form 
$ w_m \cdot e^{-j \bmk_s \cdot \bmr_m^k}$ \cite{deng2021reconfigurable}, where $w_m$ denotes the holographic weight for beamforming for the $m$-th element and $e^{-j \bmk_s \cdot \bmr_m^k}$ represents the phase of the reference wave, with $\bmk_s$ denoting the the reference wave vector and $\bmr_m^k$ denoting the distance from the $k$-th RF feed to the $m$-th element of the RHS.
 
\subsection{Communications Model}

The received signal at user $d$ can be expressed as
\begin{equation}
    y_d = \mathbf{h}_d^H \mathbf{W} \mathbf{v_d} s_d + n_d,
\end{equation}
where $\mathbf{h}_d \in \mathbb{C}^{M \times 1}$ represents the channel vector between the BS and user $d$, $s_d$ denotes the unit variance data stream, and the additive noise is denoted as $n_d \sim \mathcal{CN}(0, \sigma_n^2)$, where $\sigma_n^2$ is the noise variance.

The achievable data rate for the downlink user is given by
\begin{equation}
    R_d = \log_2\left(1 + \frac{|\mathbf{h}_d^H \mathbf{W} \mathbf{v_d}|^2}{ \sigma_n^2}\right).
\end{equation}

This model captures the interplay between the holographic beamforming matrix $\mathbf{W}$, which defines the spatial control at the analog level, and the digital beamforming matrix $\mathbf{v_d}$.  Note that in this case, the problem of rate maximization is equivalent to the maximization of $\text{Tr}(\mathbf{h}_d^H \mathbf{W} \mathbf{v_d}\mathbf{v_d}^H \mathbf{W}^H \mathbf{h}_d)$.

\subsection{Radar Sensing Model}
Furthermore, a single target is assumed to be located in the 3D sensing area at the angular direction $(\theta_t,\phi_t)$. To detect and estimate the target’s direction, the BS leverages the echoes of the downlink signal, reflected from the target and captured at the RHS. By processing these reflected signals, the BS performs radar-like sensing.
The transmitted radar signal from the BS is shaped by both the holographic and digital beamformers. Therefore, the signal reflected from the target and received at the RHS can be modelled as
\begin{equation}
    \mathbf{y}_{\text{r}} = \gamma\; \mathbf{a}(\theta_t, \phi_t) \mathbf{a}^H(\theta_t, \phi_t) \mathbf{W} \mathbf{v_d}s + \mathbf{n}_{\text{sense}},
\end{equation}
where:
\begin{itemize}
    \item $\gamma$ is the complex reflection coefficient of the target, representing its scattering properties,
    \item $\mathbf{a}(\theta_t, \phi_t) \in \mathbb{C}^{M \times 1}$ is the array response vector of the RHS for the target’s azimuth ($\theta_t$) and elevation ($\phi_t$) angles,
    \item $\mathbf{n}_{\text{sense}} \sim \mathcal{CN}(0, \sigma_{\text{r}}^2 \mathbf{I})$ is the additive white Gaussian noise.
\end{itemize}

By directly including the beamformers $\mathbf{W}$ and $\mathbf{v_d}$, the received signal explicitly highlights the influence of the holographic and digital beamforming matrices on the radar sensing process. The term $\gamma \mathbf{a}(\theta_t, \phi_t) \mathbf{a}^H(\theta_t, \phi_t) \mathbf{W} \mathbf{v_d}$ represents the directional energy reflected from the target, shaped by the hybrid holographic beamforming.

We assume the RHS to be aligned along the $x$ and $y$-axis.
The array response vector $\mathbf{a}(\theta_t, \phi_t)$ that captures the spatial signature of the signal arriving at the RHS from the target can be modelled as a uniform planner array (UPA) response. For the RHS with $\sqrt{M}$ elements along the $x$-axis and $y$-axis, its response can be modelled as
\begin{equation}
    \mathbf{a}(\theta_t, \phi_t) = \mathbf{a}_y(\theta_t, \phi_t) \otimes \mathbf{a}_x(\theta_t, \phi_t),
\end{equation}
where $\otimes$ denotes the Kronecker product. The $x$- and $y$-axis components are given as follows
\begin{subequations}
\begin{equation}
    \mathbf{a}_x(\theta_t,\phi_t) = 
    \begin{bmatrix}
        1 \\
        e^{j k_f d_x \sin\theta_t \cos\phi_t} \\
        e^{j 2 k_f d_x \sin\theta_t \cos\phi_t} \\
        \vdots \\
        e^{j (\sqrt{M}-1) k_f d_x \sin\theta_t \cos\phi_t}
    \end{bmatrix}, 
    \end{equation}
    \begin{equation}
    \mathbf{a}_y(\theta_t,\phi_t) = 
    \begin{bmatrix}
        1 \\
        e^{j k_f d_y \sin\theta_t \sin\phi_t} \\
        e^{j 2 k_f d_y \sin\theta_t \sin\phi_t} \\
        \vdots \\
        e^{j (\sqrt{M}-1) k_f d_y \sin\theta_t \sin\phi_t}
    \end{bmatrix}.
\end{equation}
\end{subequations}
where $k_f = \frac{2\pi}{\lambda}$ is the wavenumber, and $d_x \ll  \lambda/2$, $d_y  \ll \lambda/2$ denote the sub-wavelength inter-element spacings along the $x$- and $y$-axes of the reconfigurable antennas.

The received signal $\mathbf{y}_{\text{r}}$ can be processed to estimate the target's direction by analyzing the azimuth and elevation angles, which can be inferred from the phase of the UPA response $\mathbf{a}(\theta_t, \phi_t)$. However, the fundamental limits of estimation accuracy for such systems are unknown.
Namely, in contrast to the CRBs derived in the literature limited to the DoA, RHS requires a more detailed derivation for 3D sensing which depends on both the elevation and azimuth angles, which is further linked to the sub-wavelength reconfigurable antenna elements spacing along the the $x$ and $y$-axis, and hybrid holographic beamformers.  

\begin{theorem}
Consider a hybrid holographic JCAS system with reconfigurable antenna elements arranged along the $x$- and $y$-axes, with sub-wavelength spacings denoted by $d_x$ and $d_y$, respectively. Let $\theta_t$ and $\phi_t$ represent the azimuth and elevation angle of the target. The CRBs for these angles, denoted as $\text{CRB}(\theta_t)$ and $\text{CRB}(\phi_t)$, are given by

\begin{align}
    \text{CRB}(\theta_t) &= \frac{\sigma_{\text{r}}^2}{2 |\gamma|^2} \text{Tr}\left( \bmv_d^H \mathbf{W}^H \mathbf{A}_{\theta_t} \mathbf{A}_{\theta_t}^H \mathbf{W} \bmv_d \right)^{-1}, \\
    \text{CRB}(\phi_t) &= \frac{\sigma_{\text{r}}^2}{2 |\gamma|^2} \text{Tr}\left( \bmv_d^H \mathbf{W}^H \mathbf{A}_{\phi_t} \mathbf{A}_{\phi_t}^H \mathbf{W} \bmv_d \right)^{-1}.
\end{align}
where $\mathbf{A}_{\theta_t}$ and $\mathbf{A}_{\phi_t}$ represent the partial derivatives of the steering vector with respect to $\theta_t$ and $\phi_t$, respectively, and they depend the sub-wavelength spacing $d_x$ and $d_y$.
\end{theorem}
\begin{proof}
    The proof is provided in Appendix \ref{CRB_derivation}.
\end{proof}

\subsection{Multi-Objective Problem Formulation}
In contrast to the available formulation in the literature, we propose a novel weighted multi-objective problem which directly takes into account the rate and the CRB in the objective function. Such a problem can be formally stated as

\begin{subequations}
    \begin{equation}
        \max_{\mathbf{v_d},\bmW} \;\alpha   \text{Tr}\Big(\mathbf{h}_d^H \mathbf{W} \mathbf{v_d} \mathbf{v_d}^H \mathbf{W}^H \mathbf{h}_d\Big) - \beta  \Big(\text{CRB}{(\theta_t)} + \text{CRB}(\phi_t)\Big),
    \end{equation} 
    \begin{equation} \label{c1}
           \text{s.t.} \quad \textbf{Tr}\Big(\mathbf{W} \mathbf{v_d} \mathbf{v_d}^H \mathbf{W}^H \Big) \leq P_{\text{total}}.
    \end{equation}
    \begin{equation} \label{c2}
        \quad 0 < w_i < 1, \forall i.
    \end{equation}
\end{subequations}
where $\eqref{c1}$ denotes the total power constraint and $\eqref{c2}$ denotes the real-valued holographic beamforming constraint \cite{deng2021reconfigurable}, and $\alpha$ and $\beta$ represent the tunable weights to  dictating the priority of each functionality.

The proposed multi-objective problem is inherently non-convex due to several intertwined factors. First, the objective couples the maximization of the communication rate with the minimization of the CRBs, whose inverse-trace structure introduces non-linear, non-convex terms directly in the objective function. Second, the real-valued holographic beamforming constraints for each element further exacerbate this complexity by introducing additional non-convex coupling with the CRB expressions. It also leads to an unconventional problem structure which is a function of the complex variables subject to real domain constraints, necessitating novel optimization methods.  Third, the power constraint entangles the digital and holographic beamforming variables, making them difficult to decouple or linearize. Finally, balancing the rate and sensing objectives, governed by different parts of the optimization function, adds a multi-objective trade-off that complicates the search space. These intertwined factors make the problem exceedingly difficult to solve with standard optimization techniques.

\section{Proposed-Solution} \label{section_3}

In this section, we present a systematic approach for solving this challenging problem by developing a novel alternating optimization framework based on the foundations of MM. We derive the surrogate functions for the challenging inverse quadratic terms present in the CRBs expression to simplify the problem structure and enable efficient and reliable optimization at each step. By iteratively updating each variable, the proposed method effectively navigates the complex optimization landscape, achieving a robust holographic JCAS design.

\subsection{Digital Beamforming}
In this section, we assume the holographic beamformer fixed and consider the optimization of the digital beamformer.
The optimization problem with respect to $\mathbf{v_d}$ can be expressed as 
\begin{subequations}
    \begin{equation}
        \max_{\mathbf{v_d}} \; \alpha   \text{Tr}\Big(\mathbf{h}_d^H \mathbf{W} \mathbf{v_d} \mathbf{v_d}^H \mathbf{W}^H \mathbf{h}_d\Big) - \beta  \Big(\text{CRB}{(\theta_t)} + \text{CRB}{(\phi_t)}\Big),
    \end{equation} 
    \begin{equation}
          \text{s.t.} \quad \textbf{Tr}\Big(\mathbf{W} \mathbf{v_d} \mathbf{v_d}^H \mathbf{W}^H \Big) \leq P_{\text{total}}.
    \end{equation}
\end{subequations}

To solve this problem directly is extremely challenging as the objective function depends on the inverse of the digital beamformer to be optimized. To simplify the expressions while maintaining the fundamental structure, we introduce new positive semi-definite matrices that encapsulate the contributions of the steering vector derivatives and the weighting matrix. This allows us to rewrite the CRBs in a more compact form.

Let $\mathbf{B}_{\theta_t} = \mathbf{W}^H \mathbf{A}_{\theta_t} \mathbf{A}_{\theta_t}^H \mathbf{W}$ and $\mathbf{B}_{\phi_t} = \mathbf{W}^H \mathbf{A}_{\phi_t} \mathbf{A}_{\phi_t}^H \mathbf{W}$ be positive semi-definite matrices associated with the azimuth and elevation angle estimation. Then, the CRBs for the azimuth angle $\theta_t$ and elevation angle $\phi_t$ can be expressed as:

\begin{align}
    \text{CRB}(\theta_t) &= \frac{\sigma_{\text{r}}^2}{2 |\gamma|^2} \left( \mathbf{v_d}^H \mathbf{B}_{\theta_t} \mathbf{v_d} \right)^{-1}, \\
    \text{CRB}(\phi_t) &= \frac{\sigma_{\text{r}}^2}{2 |\gamma|^2} \left( \mathbf{v_d}^H \mathbf{B}_{\phi_t} \mathbf{v_d} \right)^{-1}.
\end{align}
This reformulation provides a structured interpretation of the CRB expressions and highlights only the impact of the digital beamformer on the estimation accuracy. Note the to solve the optimization problem, the non-convexity arises from the inverse dependence on $\mathbf{v_d}$ in the CRB terms. To simplify the optimization, we introduce a surrogate function for the inverse quadratic terms of the form $(\mathbf{v_d}^H \mathbf{B} \mathbf{v_d})^{-1}$. The derivation of the surrogates is based on the following more general result.
\begin{theorem} \label{theorema_2}
Let \(\mathbf{B}\) be a positive semi-definite Hermitian matrix, and let \(\mathbf{v_d} \in \mathbb{C}^R\). Any inverse quadratic term of the form \((\mathbf{v_d}^H \mathbf{B} \mathbf{v_d})^{-1}\) can be upper-bounded by the following surrogate function
\begin{align}
    \frac{1}{\mathbf{v_d}^H \mathbf{B} \mathbf{v_d}} \leq \frac{2}{\mathbf{v_d}^{(t)H} \mathbf{B} \mathbf{v_d}^{(t)}} - \frac{\mathbf{v_d}^H \mathbf{B} \mathbf{v_d}}{\left(\mathbf{v_d}^{(t)H} \mathbf{B} \mathbf{v_d}^{(t)}\right)^2},
\end{align}
where \(\mathbf{v_d}^{(t)}\) is a given point around which the surrogate is constructed.
\end{theorem}
\begin{proof}
    The proof is provided in Appendix \ref{proof_Thm2}
\end{proof}

Using the results provided in the theorem, we can construct the majorizers for the multi-objective function with the surrogates. Let $S_{\theta_t}$ and $ S_{\phi_t}$ denote the surrogates of the CRBs $\text{CRB}{(\theta_t)}$ and $\text{CRB}{(\phi_t)}$, respectively. They can be written as
\begin{align}
    S_{\theta_t} &= \frac{2 \sigma_{\text{r}}^2}{2 |\gamma|^2 (\mathbf{v_d}^{(t)H} \mathbf{B}_{\theta_t} \mathbf{v_d}^{(t)})} - \frac{\sigma_{\text{r}}^2 \mathbf{v_d}^H \mathbf{B}_{\theta_t} \mathbf{v_d}}{2 |\gamma|^2 \left(\mathbf{v_d}^{(t)H} \mathbf{B}_{\theta_t} \mathbf{v_d}^{(t)}\right)^2}, \\
    S_{\phi_t} &= \frac{2 \sigma_{\text{r}}^2}{2 |\gamma|^2 (\mathbf{v_d}^{(t)H} \mathbf{B}_{\phi_t} \mathbf{v_d}^{(t)})} - \frac{\sigma_{\text{r}}^2 \mathbf{v_d}^H \mathbf{B}_{\phi_t} \mathbf{v_d}}{2 |\gamma|^2 \left(\mathbf{v_d}^{(t)H} \mathbf{B}_{\phi_t} \mathbf{v_d}^{(t)}\right)^2}.
\end{align}
By using them we can restate the original multi-objective function at iteration $t$ as follows
\begin{align}
    Q(\mathbf{v_d}, \mathbf{v_d}^{(t)}) = \alpha   \text{Tr}\Big(\mathbf{h}_d^H \mathbf{W} \mathbf{v_d} \mathbf{v_d}^H \mathbf{W}^H \mathbf{h}_d\Big)  - \beta   \Big(S_{\theta_t} + S_{\phi_t}\Big).
\end{align}
Note that the majorized function $Q(\mathbf{v_d}, \mathbf{v_d}^{(t)})$ is quadratic in $\mathbf{v_d}$, thus is can enable a simplified and tacktable optimization of the digital beamformer $\bmv_d$. Once the original function is majorized, the goal is to maximize the optimal digital beamformer with an iterative update while adhering to the total transmit power constraint. Namely, at each iteration \( t \), the optimization problem can be reformulated as

\begin{equation}
\mathbf{v_d}^{(t+1)} =  \max_{\textbf{Tr}\Big(\mathbf{W} \mathbf{v_d} \mathbf{v_d}^H \mathbf{W}^H \Big) \leq P_{\text{total}}.} Q(\mathbf{v_d}, \mathbf{v_d}^{(t)}).
\end{equation}

Let \( \mathbf{M}^{(t)} \) denote the effective majorization matrix at iteration \( t \), which encodes the optimization structure based on all the parameters and constraints of the system for holographic JCAS. The explicit form of \( \mathbf{M}^{(t)} \) is given by

\begin{equation}
\mathbf{M}^{(t)} = \alpha \mathbf{W}^H \mathbf{h}_d \mathbf{h}_d^H \mathbf{W} - \beta \cdot \frac{\sigma_{\text{r}}^2}{2 |\gamma|^2} \mathbf{B}_{\text{eff}}^{(t)},
\end{equation}
where the effective matrix \( \mathbf{B}_{\text{eff}}^{(t)} \) contains terms that vary in time and is defined as a weighted combination of two matrices obtained from the surrogates of components related to the parameters \( \theta_t \) and \( \phi_t \), and is given by

\begin{equation}
\mathbf{B}_{\text{eff}}^{(t)} = -\frac{\mathbf{B}_{\theta_t}}{\left(\mathbf{v_d}^{(t)H} \mathbf{B}_{\theta_t} \mathbf{v_d}^{(t)}\right)^2} - \frac{\mathbf{B}_{\phi_t}}{\left(\mathbf{v_d}^{(t)H} \mathbf{B}_{\phi_t} \mathbf{v_d}^{(t)}\right)^2}.
\end{equation}

To make the optimization tractable, the majorization matrix is substituted into the problem. This reformulation leads to a quadratic optimization problem which can be expressed as

\begin{equation}
\mathbf{v_d}^{(t+1)} =  \max_{\textbf{Tr}\Big(\mathbf{W} \mathbf{v_d} \mathbf{v_d}^H \mathbf{W}^H \Big) \leq P_{\text{total}}.} \mathbf{v_d}^H \mathbf{M}^{(t)} \mathbf{v_d}.
\end{equation}

This quadratic optimization problem seeks to maximize the quadratic form \(\mathbf{v_d}^H \mathbf{M}^{(t)} \mathbf{v_d}\) under the given power constraint. This problem has a well-known solution: the optimal \(\mathbf{v_d}^{(t+1)}\) is the eigenvector corresponding to the largest eigenvalue of the majorization matrix \(\mathbf{M}^{(t)}\), denoted as $\mathbf{e}_{\max}(\mathbf{M}^{(t)})$.  Once the optimal solution is obtained, it can be scaled as follows to meet the sum-power constraint   
\begin{align} \label{update_dig}
    \mathbf{v_d}^{(t+1)} = \sqrt{P_{\text{total}}} \cdot \frac{\mathbf{e}_{\max}(\mathbf{M}^{(t)})}{|\mathbf{W} \mathbf{e}_{\max}(\mathbf{M}^{(t)})|_F}.
\end{align}

\subsection{Holographic Beamforming}  
 The holographic beamformer $\mathbf{W}$ maps the RF chains to the $M$ antenna elements, enabling joint optimization of the communication and sensing functions. In this section, we consider its optimization assuming the digital beamformer to be fixed.
 
To do so, we first aim at highlighting the dependence of the $\bmW$ on the tunable holographic weights. Namely, the holographic beamformer can be expressed as 
\begin{equation} \label{holo_structure}
    \mathbf{W} = \text{diag}(\mathbf{w})\mathbf{\Phi}, 
\end{equation}
where:
\begin{itemize}
    \item $\mathbf{w} \in \mathbb{R}^{M \times 1}$ is the vector of tunable amplitude coefficients, constrained as $0 \leq w_i \leq 1$.
    \item $\mathbf{\Phi} \in \mathbb{C}^{M \times R}$ is a fixed phase-shift matrix with elements of the form $ e^{-j \bmk_s \cdot \bmr_m^k}$.
\end{itemize}

Therefore, the optimization problem with respect to holographic beamformer can be formally stated as 
\begin{subequations}
\label{eq:main_problem}
\begin{equation} \label{O_holo}
    \max_{\mathbf{w}} \alpha   \text{Tr}\Big(\mathbf{h}_d^H \mathbf{W} \mathbf{v_d} \mathbf{v_d}^H \mathbf{W}^H \mathbf{h}_d\Big) - \beta  \Big(\text{CRB}{(\theta_t)} + \text{CRB}{(\phi_t)}\Big)
\end{equation}
\begin{equation}
    \text{s.t.} \quad 0 \leq w_i \leq 1, \quad \forall i.
\end{equation}
\end{subequations}

Note that each element of the holographic beamformer does not amplify further the power transmitted from the digital beamforming, as its each element assumed values in the range $[0,1]$. Therefore the constraint $\textbf{Tr}\Big(\mathbf{W} \mathbf{v_d} \mathbf{v_d}^H \mathbf{W}^H \Big) \leq P_{\text{total}}$ for the optimization of $\bmw$ can be omitted. To derive a tractable solution for $\bmW$, we first consider rewriting each terms in the objective function by highlighting its dependence on the holographic weights $\bmw$. The first term in \eqref{O_holo} quantifies the total effective signal power at the intended user. It can be expressed as a function of $\bmw$ as

\begin{equation}
    \text{Tr}(\mathbf{h}_d^H \text{diag}(\mathbf{w})\mathbf{\Phi} \mathbf{v}_d \mathbf{v}_d^H  \mathbf{\Phi}^H \text{diag}(\mathbf{w})^H \mathbf{h}_d).
\end{equation}
and by using the cyclic property of the trace operator, we can arrange it as

\begin{equation}
    \text{Tr}(\text{diag}(\mathbf{w})^H \mathbf{\Phi} \mathbf{v}_d \mathbf{v}_d^H \mathbf{\Phi}^H \text{diag}(\mathbf{w}) \mathbf{h}_d \mathbf{h}_d^H).
\end{equation}

By using the property of the $\text{diag}(\cdot)$ and the properties of the Hadamard product, we can define the following matrix 
 
\begin{equation}
    \mathbf{Q}_{\text{c}} = \text{diag}(\mathbf{\Phi} \mathbf{v}_d \mathbf{v}_d^H \mathbf{\Phi}^H) \odot (\mathbf{h}_d \mathbf{h}_d^H).
\end{equation}
which allows us to rewrite the received signal power at the user in a quadratic form of $\bmw$ as 

\begin{equation}
    \mathbf{w}^H \mathbf{Q}_{\text{c}} \mathbf{w}.
\end{equation}

 Subsequently, our objective is to express the CRBs in a similar form. To do so, first note that based on the observation \eqref{holo_structure}, and using the properties of the trace operator, we can rewrite the CBRs as

%\begin{equation}
  %  \text{CRB}(\theta_t) = \frac{\sigma_{\text{r}}^2}{2 |\gamma|^2} 
  %  \text{Tr} \left( \mathbf{v}_d^H \mathbf{\Phi}^H \text{diag}(\mathbf{w})^H \mathbf{A}_{\theta_t} \mathbf{A}_{\theta_t}^H \text{diag}(\mathbf{w}) \mathbf{\Phi} \mathbf{v}_d \right)^{-1},
%\end{equation}

%\begin{equation}
 %   \text{CRB}(\phi_t) = \frac{\sigma_{\text{r}}^2}{2 |\gamma|^2} 
  %  \text{Tr} \left( \mathbf{v}_d^H %\mathbf{\Phi}^H \text{diag}(\mathbf{w})^H \mathbf{A}_{\phi_t} \mathbf{A}_{\phi_t}^H \text{diag}(\mathbf{w}) \mathbf{\Phi} \mathbf{v}_d \right)^{-1}.
%\end{equation}

%Using the cyclic property of the trace operator, we move \( \mathbf{v}_d^H \mathbf{\Phi}^H \) to the right-hand side:

\begin{equation}
    \text{CRB}(\theta_t) \hspace{-1mm}= \hspace{-1mm}\frac{\sigma_{\text{r}}^2}{2 |\gamma|^2} 
    \text{Tr} \left( \text{diag}(\mathbf{w})^H \mathbf{A}_{\theta_t} \mathbf{A}_{\theta_t}^H \text{diag}(\mathbf{w}) \mathbf{\Phi} \mathbf{v}_d \mathbf{v}_d^H \mathbf{\Phi}^H \right)^{-1}
\end{equation}

\begin{equation}
    \text{CRB}(\phi_t)\hspace{-1mm} = \hspace{-1mm}\frac{\sigma_{\text{r}}^2}{2 |\gamma|^2} 
    \text{Tr} \left( \text{diag}(\mathbf{w})^H \mathbf{A}_{\phi_t} \mathbf{A}_{\phi_t}^H \text{diag}(\mathbf{w}) \mathbf{\Phi} \mathbf{v}_d \mathbf{v}_d^H \mathbf{\Phi}^H \right)^{-1}
\end{equation}

To further simplify the CRB expressions and express them in terms of the holographic weights $\mathbf{w}$, we utilize the properties of the trace operator, diagonal operator, and Hadmarad product, which allow us to define the followings matrices

\begin{equation}
    \mathbf{Q}_{\theta_t} = \text{diag}(\mathbf{A}_{\theta_t} \mathbf{A}_{\theta_t}^H) \odot (\mathbf{\Phi} \mathbf{v}_d \mathbf{v}_d^H \mathbf{\Phi}^H),
\end{equation}

\begin{equation}
    \mathbf{Q}_{\phi_t} = \text{diag}(\mathbf{A}_{\phi_t} \mathbf{A}_{\phi_t}^H) \odot (\mathbf{\Phi} \mathbf{v}_d \mathbf{v}_d^H \mathbf{\Phi}^H),
\end{equation}
and express the CRBs as a function of the holographic weights as

\begin{equation}
    \text{CRB}(\theta_t) = \frac{\sigma_{\text{r}}^2}{2 |\gamma|^2} 
    \left( \mathbf{w}^H \mathbf{Q}_{\theta_t} \mathbf{w} \right)^{-1},
\end{equation}

\begin{equation}
    \text{CRB}(\phi_t) = \frac{\sigma_{\text{r}}^2}{2 |\gamma|^2} 
    \left( \mathbf{w}^H \mathbf{Q}_{\phi_t} \mathbf{w} \right)^{-1}.
\end{equation}
Thus, we have successfully expressed the CRB in the desired quadratic form $\mathbf{w}^H \mathbf{Q} \mathbf{w}$, where $\mathbf{Q}_{\theta_t}$ and $\mathbf{Q}_{\phi_t}$ incorporate the spatial structure of the sensing matrices and the impact of the digital beamforming vector $\mathbf{v}_d$.

The holographic beamforming optimization problem with respect to the holographic weights can be formally restated as 
\begin{subequations}
\label{eq:final_optimization_problem}
\begin{equation}
    \begin{aligned}
         \max_{\mathbf{w}} \quad  \alpha \mathbf{w}^H \mathbf{Q}_{\text{c}} \mathbf{w} 
         & - \beta  \frac{\sigma_r^2}{2|\gamma|^2} \left( \mathbf{w}^H \mathbf{Q}_{\theta_t} \mathbf{w} \right)^{-1} \\
         & - \beta \frac{\sigma_r^2}{2|\gamma|^2} \left( \mathbf{w}^H \mathbf{Q}_{\phi_t} \mathbf{w} \right)^{-1}
    \end{aligned}
\end{equation}
   \begin{equation}
       \text{s.t.} \quad  0 \leq w_i \leq 1, \forall i.
   \end{equation}
\end{subequations}
To efficiently solve this highly non-convex problem, we first develop the surrogates to majorize the objective function in a tractable form.  
 \begin{lemma}
For any real-valued diagonal matrix \( \mathbf{D} \succ 0 \) and a feasible point \( \mathbf{w}^{(t)}\in \mathbb{R}^M \), the following inequality holds
\begin{equation}
    \frac{1}{\mathbf{w}^H \mathbf{D} \mathbf{w}} \leq 
    \frac{1}{\mathbf{w}^{(t)H} \mathbf{D} \mathbf{w}^{(t)}} 
    + \frac{\mathbf{w}^{(t)H} \mathbf{D} \mathbf{w}^{(t)} 
    - 2 \mathbf{w}^H \mathbf{D} \mathbf{w}^{(t)}}
    {\left(\mathbf{w}^{(t)H} \mathbf{D} \mathbf{w}^{(t)}\right)^2}.
\end{equation}
\end{lemma}
\begin{proof}
    The proof follows directly by applying the results in Theorem $2$.
\end{proof}
Applying Lemma 1 to the CRB terms, we obtain the majorized objective function
\begin{equation}
\begin{aligned}
    \widetilde{f}(\mathbf{w}|\mathbf{w}^{(t)}) &= \alpha \mathbf{w}^H \mathbf{D}_{\text{c}} \mathbf{w} \\
    & - \beta \frac{\sigma_r^2}{2|\gamma|^2}  
    \Bigg[ 
    \frac{2}{\mathbf{w}^{(t)H} \mathbf{Q}_{\theta_t} \mathbf{w}^{(t)}} 
    - \frac{\mathbf{w}^H \mathbf{Q}_{\theta_t} \mathbf{w}}
    {\left(\mathbf{w}^{(t)H} \mathbf{Q}_{\theta_t} \mathbf{w}^{(t)}\right)^2}
    \Bigg]
    \\
    & - \beta \frac{\sigma_r^2}{2|\gamma|^2}
    \Bigg[ 
    \frac{2}{\mathbf{w}^{(t)H} \mathbf{Q}_{\phi_t} \mathbf{w}^{(t)}} 
    - \frac{\mathbf{w}^H \mathbf{Q}_{\phi_t} \mathbf{w}}
    {\left(\mathbf{w}^{(t)H} \mathbf{Q}_{\phi_t} \mathbf{w}^{(t)}\right)^2}
    \Bigg].
\end{aligned}
\end{equation}

Thus, the optimization problem at each iteration becomes

\begin{subequations}
\label{eq:simplified_problem_mm}
\begin{align}
    \max_{\mathbf{w}} \quad &  \widetilde{f}(\mathbf{w}|\mathbf{w}^{(t)})  \\
    \text{s.t.} \quad & 0 \leq w_i \leq 1, 
\end{align}
\end{subequations}

Let $\mathbf{M}^{(t)} $ denote the majorization matrix which captures the dependence only on the terms varying in time, given as   
 \begin{equation}
\begin{aligned}
    \mathbf{M}^{(t)} \triangleq \alpha \mathbf{D}_{\text{comm}} 
    & - \frac{\beta\sigma_r^2}{2|\gamma|^2}  
    \left( \frac{\mathbf{Q}_{\theta_t}}{\left(\mathbf{w}^{(t)H} \mathbf{Q}_{\theta_t} \mathbf{w}^{(t)}\right)^2} \right) \\
    & - \frac{\beta\sigma_r^2}{2|\gamma|^2}  
    \left( \frac{\mathbf{Q}_{\phi_t}}{\left(\mathbf{w}^{(t)H} \mathbf{Q}_{\phi_t} \mathbf{w}^{(t)}\right)^2} \right).
\end{aligned}
\label{eq:M_matrix}
\end{equation}
then we can restate the optimization problem at iteration $t$ in a simplified form as
\begin{subequations}
\label{eq:convex_subproblem}
\begin{align}
    \max_{\mathbf{w}} \quad & \mathbf{w}^H \mathbf{M}^{(t)} \mathbf{w} \\
    \text{s.t.} \quad & 0 \leq w_i \leq 1. 
\end{align}
\end{subequations}

The optimization problem aims to maximize the quadratic function \( f(\mathbf{w}) = \mathbf{w}^H \mathbf{M}^{(t)} \mathbf{w} \) under the box constraint \( 0 \leq w_i \leq 1 \). To solve this problem, we employ a PGA approach, iteratively updating \( \mathbf{w} \) in the direction of the gradient while ensuring feasibility. Let $\nabla f(\mathbf{w})$ denote the gradient of the objective function, given as
\begin{equation}
    \nabla f(\mathbf{w}) = 2 \mathbf{M}^{(t)} \mathbf{w}.
\end{equation}
Using this gradient, we perform an ascent step at each iteration to update the holographic weights with the following update rule
\begin{equation}
    \mathbf{w}^{(k+1)} = \mathbf{w}^{(k)} + \eta \nabla f(\mathbf{w}^{(k)}),
\end{equation}
where \( \eta \) is the learning step size. Since \( \mathbf{w} \) is subject to the box constraint, we apply the following projection after each update on each element  
\begin{equation}
   \mathbf{w} \leftarrow \min(1, \max(0, \mathbf{w})).
\end{equation}
This projection ensures that \( \mathbf{w} \) remains within the feasible set throughout the iterations. The optimization process can be repeated iteratively until convergence, which is determined by monitoring the relative change in \( \mathbf{w} \) between successive iterations. Since we are tackling a multi-objective problem, we impose the convergence condition such that the variation in the communication rate and in both the CRBs are below a predefined threshold or a maximum number iterations $T_{max}$ are reached. The iterative procedure to update the holographic beamformer is formally stated in Algorithm $1$.

 \begin{algorithm}[t]
\caption{PGA-Based Holographic Beamforming Optimization}
\label{alg:gradient_beamforming}
\begin{algorithmic}[1]
\State \textbf{Input:} Initial holographic beamforming vector $\mathbf{w}^{(0)}$, learning step size $\eta_0$, maximum iterations $T_{\max}^{(w)}$, and convergence tolerance $\epsilon$.
\State \textbf{Initialize:} Set iteration counter $t = 0$, initial step size $\eta \gets \eta_0$, compute $\mathbf{D}_a = \text{diag}(|\mathbf{\Phi} \mathbf{v}_d|^2)$.
\Repeat
    \State Compute the matrix $\mathbf{M}^{(t)}$

    \State Compute the gradient as 
   $ \nabla_{\mathbf{w}} f(\mathbf{w}) = 2 \mathbf{M}^{(t)} \mathbf{w}$
    \State Update the holographic beamformer as
    \begin{equation*}
        \mathbf{w}_{\text{new}} = \mathbf{w}^{(t)} + \eta \cdot \text{Re}(\nabla_{\mathbf{w}} f(\mathbf{w}^{(t)}))
    \end{equation*}
    \State Project onto feasible set as
    \begin{equation*}
        \mathbf{w} \gets \min(1, \max(0, \mathbf{w}_{\text{new}}))
    \end{equation*}
    \State Increment iteration counter: $t \gets t + 1$.
\Until{Convergence or maximum iterations reached}
\State \textbf{Output:} Optimized $\bmW = \mathbf{w} \mathbf{\Phi}$.
\end{algorithmic}
\end{algorithm}

The final algorithm which jointly optimize the digital beamforming and holographic beamforming to optimize the multi-objective function efficiently is formally given in Algorithm $2$.

\begin{algorithm}[t]
\caption{MM-Based Hybrid Holographic JCAS Optimization}
\label{alg:gradient_beamforming}
\begin{algorithmic}[1]
\State \textbf{Input:} Initial digital beamforming vector $\mathbf{v_d}^{(0)}$, initial holographic weights $\mathbf{w}^{(0)}$, maximum iterations $T_{\max}$, convergence tolerance $\epsilon$.
\State \textbf{Initialize:} Set iteration counter $t = 0$.
\Repeat
    \State Compute $\mathbf{B}_{\text{eff}}^{(t)}$ based on current $\mathbf{v_d}^{(t)}$.
    \State Determine the matrix $\mathbf{M}^{(t)}$ using $\alpha$, $\beta$, and $\mathbf{B}_{\text{eff}}^{(t)}$.
    \State Extract the principal eigenvector of $\mathbf{M}^{(t)}$
    \State Update $\mathbf{v_d}^{(t+1)}$ as \eqref{update_dig}.
    \State Update Holographic Beamforming with Algorithm 1.
    \State Increment iteration counter: $t \gets t + 1$.
\Until{Convergence or maximum iterations reached}
\State \textbf{Output:} Optimized  $\mathbf{v_d}$ and $\mathbf{W}$.
\end{algorithmic}
\end{algorithm}

\subsection{Proof of Convergence}

To validate the effectiveness and reliability of the proposed MM-based iterative optimization algorithm, we present a mathematical proof demonstrating its convergence properties. This proof leverages the foundational principles of the MM framework, ensuring that each iterative step consistently improves the objective function and ultimately converges to a stationary point of the original non-convex optimization problem.

The optimization problem aims to maximize a weighted combination of the communication rate and the minimization of the CRBs. The MM framework facilitates the optimization of such a complex non-convex function by iteratively solving simpler surrogate problems. At each iteration, a surrogate function that upper-bounds the original objective function is constructed. This surrogate must satisfy two essential conditions
\begin{enumerate}
    \item \textbf{Majorization:} The surrogate function must lie above the original objective function for all feasible points.
    \item \textbf{Tangent Condition:} The surrogate function must coincide with the original objective function at the current iteration.
\end{enumerate}

By maximizing this surrogate function, the algorithm ensures that the original objective function does not decrease with each iteration, promoting a monotonic improvement towards an optimal solution. To address the non-convex inverse quadratic terms in the CRB expressions, we introduce surrogate functions derived from a first-order Taylor expansion. These surrogates effectively upper-bound the inverse terms, transforming the original objective into a quadratic form that is more tractable for optimization. Specifically, the surrogate function replaces the challenging inverse terms with linear approximations around the current iterate, ensuring that the majorization and tangent conditions are satisfied. At each iteration, the algorithm updates the beamforming vectors by maximizing the surrogate function. Due to the construction of the surrogate, this step guarantees that the value of the original objective function either increases or remains unchanged. Consequently, the sequence of objective function values generated by the algorithm leads to a monotonic improvement. Given that the objective function is bounded above by system constraints, this monotonicity ensures convergence. 
\begin{figure}
    \centering
    \includegraphics[width=0.7\linewidth]{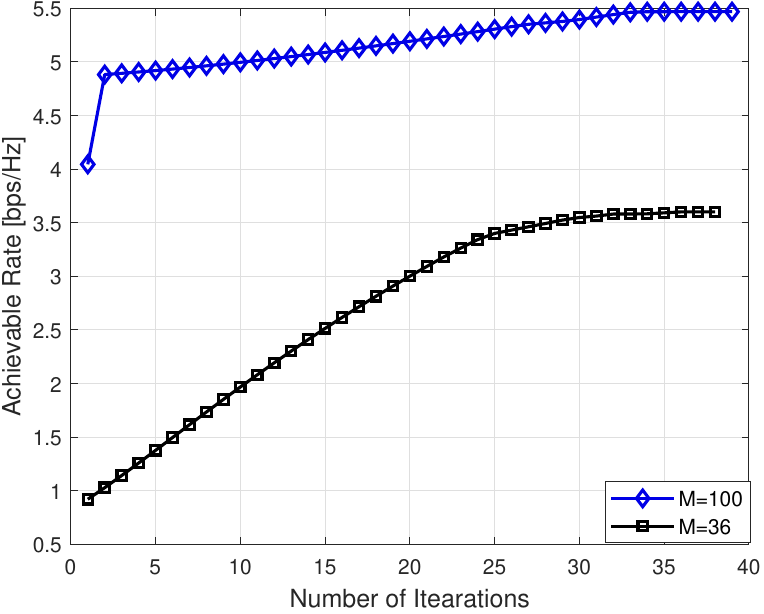}
    \caption{Typical convergence behaviour of the proposed algorithm for the communications rate.}
    \label{conv_rate}
\end{figure}
\begin{figure}
    \centering
    \includegraphics[width=0.7\linewidth]{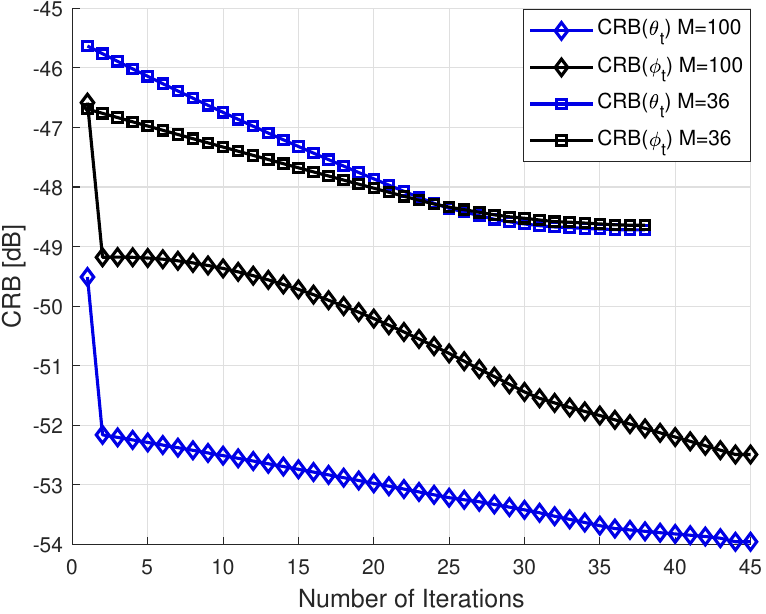}
    \caption{Typical convergence behaviour of the proposed algorithm for the CRBs of elevation and azimuth angles.}
    \label{conv_CRB}
\end{figure}

 \begin{figure*}
    \centering
 \begin{minipage}{0.48\textwidth}
      \centering
    \includegraphics[width=0.7\linewidth]{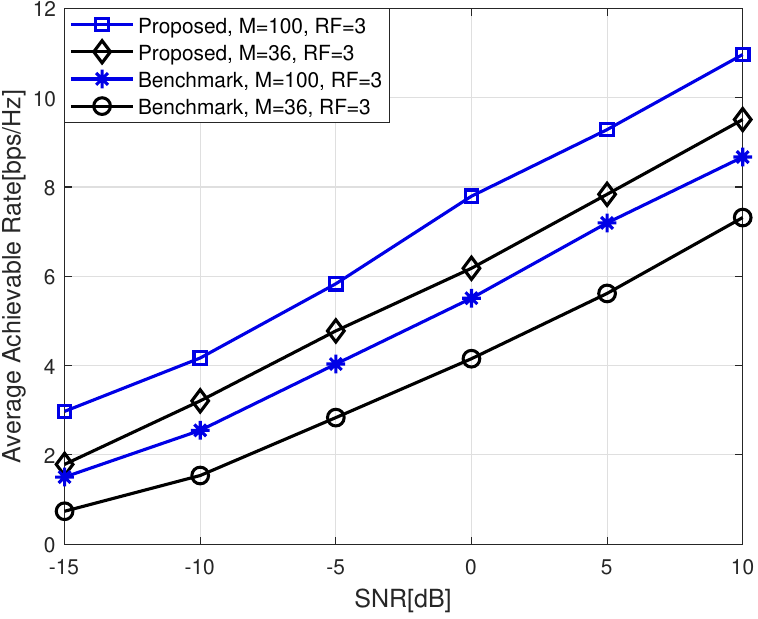}
    \caption{Average communication rate as a function of the SNR with $3$ RF chains.}
    \label{rate_3RF}
\end{minipage}  
      \begin{minipage}{0.48\textwidth}
        \centering
    \includegraphics[width=0.7\linewidth]{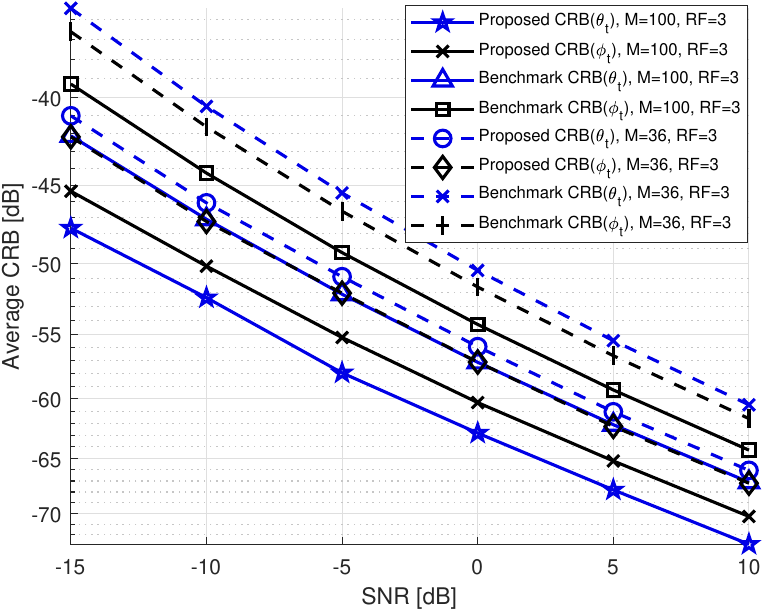}
    \caption{Average CRB as a function of the SNR with $3$ RF chains.}
    \label{CRB_3RF}
    \end{minipage}  
\end{figure*}

 A typical convergence behaviour of the multi-objective function for the rate maximization and CRBs minimization with an RHS of size $36$ and $100$ is presented in Figures \ref{conv_rate} and \ref{conv_CRB}, respectively.

\subsection{Complexity Analysis}

The digital beamforming update involves constructing the majorization matrix \( \mathbf{M}^{(t)} \in \mathbb{C}^{R \times R} \), which requires matrix multiplications with a complexity of \(  \mathcal{O}(R^2 M) \). The dominant eigenvector extraction via eigenvalue decomposition contributes an additional complexity of \(  \mathcal{O}(R^3) \), leading to a total complexity per iteration for digital beamforming of $ \mathcal{O}(R^2 M + R^3).$
For the holographic beamformer update, the optimization problem is solved using PGA, where the gradient calculation involves the matrix-vector multiplication \( \nabla_{\mathbf{w}} f(\mathbf{w}) = 2 \mathbf{M}^{(t)} \mathbf{w} \), contributing a complexity of \(  \mathcal{O}(M^2) \). Furthermore, the projection of each element onto the feasible set adds further complexity of \(  \mathcal{O}(M) \). Since the PGA method updates the holographic beamformer until convergence, the number of iterations required denoted as \( T_{\max}^{(w)} \), leads to a total complexity for the holographic beamforming update of $\mathcal{O}(T_{\max}^{(w)} M^2 + M)$.
Since the overall algorithm alternates between updating the digital beamforming vector and the holographic beamformer for a maximum of \( T_{\max} \) iterations, the total complexity of the proposed MM-based optimization framework is given by

\begin{equation}
    O\left( T_{\max} \left( R^2 M + R^3 + T_{\max}^{(w)} M^2  + M\right) \right).
\end{equation}

\section{Simulation Results} \label{section_4}
In this section, we present the simulation results to investigate the performance of holographic JCAS. We define the signal-to-noise ratio (SNR) of the system as $P_{\text{total}}/\sigma_n^2$ and it is assumed that $\sigma_n^2 = \sigma_r^2$. The transmit power $P_{\text{total}}$ is set to $1$ and the noise variance is chosen to meet the desirable SNR. The center of the RHS is assumed at the position $(0,0,0)$ and the system is assumed to be operating at the frequency of $f=20$GHz with element spacing of  \(\lambda/4\) for the RHS elements. The target is assumed to be located in the direction of $(45^\circ,60^\circ)$ and the communication channels between the RHS and the
user is Rayleigh distributed, i.e., the each element of $\bmh$ follows a standard complex Gaussian distribution with mean $0$ and
variance $1$, as in \cite{zhang2022holographic}. The substrate of the parallel-plate waveguide has a refractive index of \( \sqrt{3} \), which determines the magnitude of the propagation vector. This can be expressed as:
$|\mathbf{k}_m| = \sqrt{3} |\mathbf{k}_f(\theta_t, \phi_t)| = \frac{2\pi \sqrt{3} f}{v_c}$
where \( v_c \) denotes the speed of light \cite{zhang2022holographic}. The holographic beamformer optimization is assumed to employ iterative amplitude optimization with a convergence threshold of \(\epsilon = 10^{-5}\) and a learning rate of \(\eta = 0.01\). The weights for the multi-objective function are set as $\alpha,\beta$ and the scattering coefficient $|\gamma|$ are set to $1$.
 
We define the benchmark scheme to compare the performance of the proposed design as such that it optimizes the digital beamformer with generalized dominant eigenvector solution and the holographic weights are randomly selected in the interval $[0,1]$. The results are plotted by averaging over $100$ channel realizations.

\begin{figure*}
    \centering
 \begin{minipage}{0.48\textwidth}
      \centering
    \includegraphics[width=0.7\linewidth]{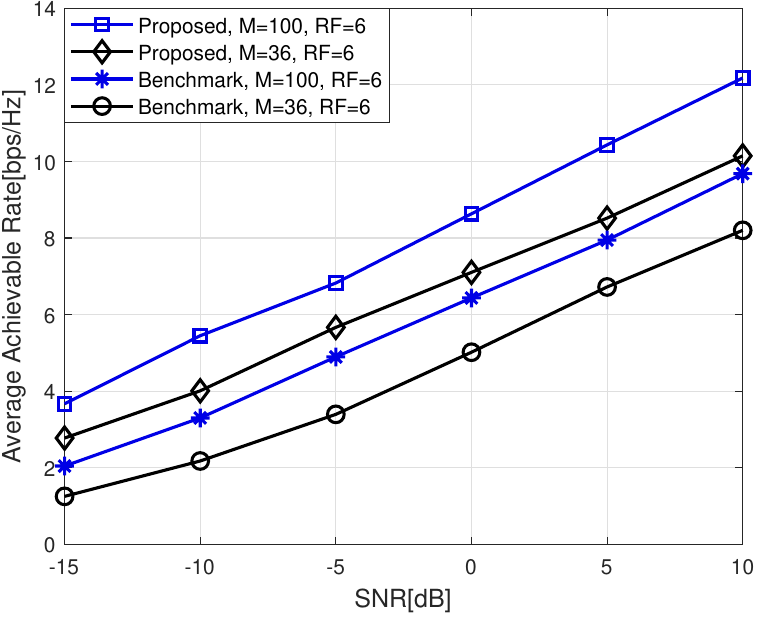}
    \caption{Average communication rate as a function of the SNR with $6$ RF chains.}
    \label{rate_6RF}
\end{minipage}  
      \begin{minipage}{0.48\textwidth}
        \centering
    \includegraphics[width=0.7\linewidth]{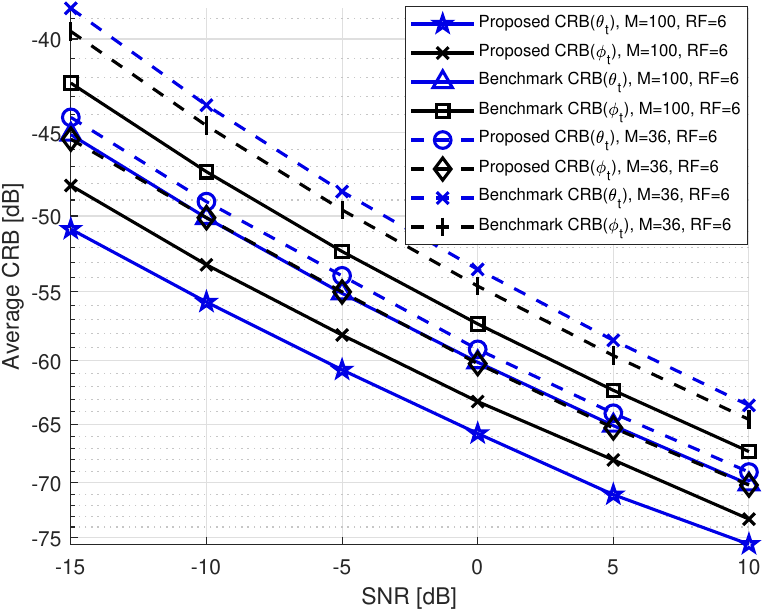}
    \caption{Average CRB as a function of the SNR with $6$ RF chains.}
    \label{CRB_6RF}
    \end{minipage}  
\end{figure*} 
 
\begin{figure*}
    \centering
 \begin{minipage}{0.48\textwidth}
         \centering
    \includegraphics[width=0.7\linewidth]{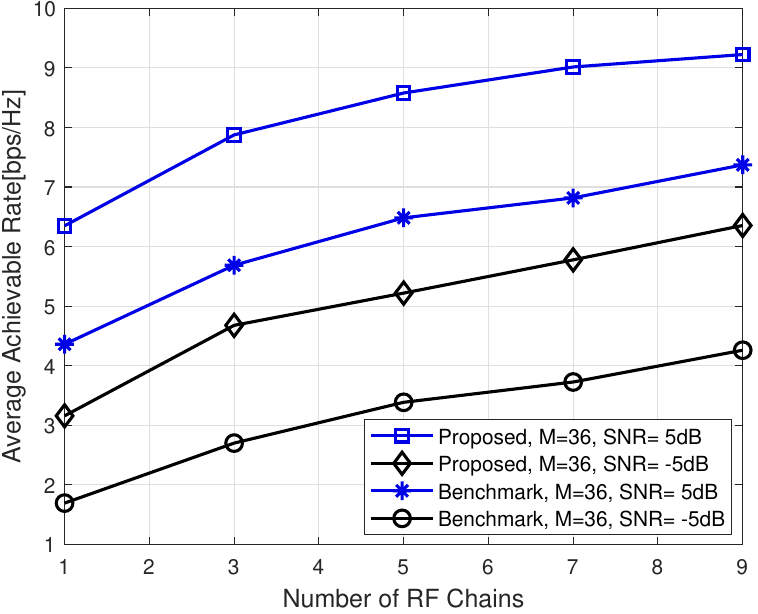}
    \caption{Average rate as a function for the number of RF chains with $M=36$.}
    \label{rate_rf_chains}
\end{minipage}  
      \begin{minipage}{0.48\textwidth}
        \centering
    \includegraphics[width=0.7\linewidth]{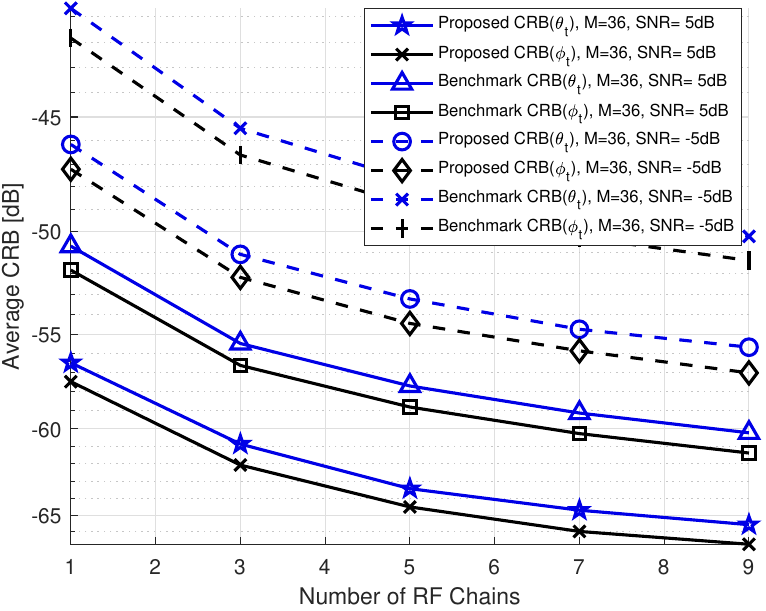}
    \caption{Average CRB as a function for the number of RF chains with $M=36$}
    \label{crb_rf_chains}
    \end{minipage}  
\end{figure*} 
Figure \ref{rate_3RF} presents the average achievable rate in bits per second per Hertz (bps/Hz) as a function of SNR for the proposed and benchmark schemes, with \( M = 100 \) and \( M = 36 \), and the number of RF chains \( R = 3 \). We can see that the proposed scheme consistently outperforms the benchmark across the entire SNR range, demonstrating its superior spectral efficiency. Notably, the performance gap between the proposed and benchmark schemes becomes more pronounced at higher SNR levels, indicating that the proposed method is particularly effective in high-SNR regimes. Furthermore, increasing the number of RHS elements \( M \) from $36$ to $100$ results in a significant improvement in the achievable rate for both schemes, highlighting the importance of a larger array size in enhancing the spectral efficiency. Figure \ref{CRB_3RF} presents the average CRB in dB as a function of the SNR for the estimation of angles \(\theta_t\) and \(\phi_t\). The results are shown for both the proposed and benchmark schemes, with two RHS element configurations: \(M = 100\) and \(M = 36\), while the number of RF chains is fixed at \(RF = 3\). It is shown that with different system parameters, the proposed scheme consistently achieves a lower CRB compared to the benchmark across the entire SNR range, demonstrating its superior performance in parameter estimation. Additionally, it is shown that increasing the RHS size results in a significant reduction in the CRB for both schemes, highlighting the advantage of larger holographic arrays in improving estimation accuracy as well. This is attributed to the increased spatial degrees of freedom, which enhances the resolution and precision of angle estimation.

Figure \ref{rate_6RF} shows the average rate as a function of the SNR for both the proposed and benchmark schemes, with \(M = 100\) and \(M = 36\), with \(RF = 6\). Compared to Figure \ref{rate_3RF}, we can clearly see the impact of doubling the number of RF chains to $6$, leading to enhanced achievable rate. The proposed scheme continues to outperform the benchmark across the entire SNR range, with the performance gap between the two schemes becoming more pronounced at all the SNR levels compared to the case of $3$ RF chains. This indicates that the proposed method is more effective in utilizing the additional degrees of freedom offered by the enhanced number of RF chains to enhance the communication rate. Furthermore, we can see that the gain is much more evident for the case of large RHS. This observation suggest that increasing the number of RF chains combined with larger RHS can lead to substantial performance improvements for communication. 
  
A similar conclusion can be drawn also for the sensing case. Namely, Figure \ref{CRB_6RF} shows the average CRBs as a function of the SNR for the estimation of angles \(\theta_t\) and \(\phi_t\). Compared to the Figure \ref{CRB_3RF}, it demonstrates the impact of doubling the number of RF chains to $6$. We can clearly see that the proposed scheme continues to achieve a lower CRB compared to the benchmark across the entire SNR range, indicating superior estimation accuracy from the case presented in Figure \ref{CRB_3RF}. We can see that increasing the number of RF chains pushes further the CRBs for 3D sensing for both the proposed and benchmark schemes, highlighting the importance of RF chain count in enhancing estimation accuracy. This improvement is particularly notable for larger antenna arrays (\(M = 100\)), where the additional RF chains provide more spatial degrees of freedom, enabling better resolution and more precise estimation of the angles \(\theta_t\) and \(\phi_t\). Even with \(M = 36\), the increase in RF chains leads to a noticeable reduction in the CRB, emphasizing the role of RF chains in improving further the estimation performance. These observation indicates that increasing the number of RF chains combined with a larger RHS can be a key factor in achieving jointly better parameter estimation and communications rate.

Figure \ref{rate_rf_chains} illustrates the average achievable rate  as a function of the number of RF chains for both the proposed and benchmark schemes, with an RHS of size \(M = 36\). The results are shown for two different SNR levels: $5$ dB and $-5$ dB. We can see that proposed scheme consistently outperforms the benchmark across all numbers of RF chains and at both SNR levels, demonstrating its superior spectral efficiency. At SNR of 5 dB, the achievable rate for both schemes increases significantly as the number of RF chains grows, with the proposed scheme showing a steeper improvement, highlighting its efficiency in leveraging additional RF chains. At a lower SNR of -5 dB, the achievable rate is generally lower for both schemes due to the increased impact of noise, but the proposed scheme still maintains a clear advantage over the benchmark, particularly as the number of RF chains increases.
\begin{figure}
    \centering
    \includegraphics[width=0.7\linewidth]{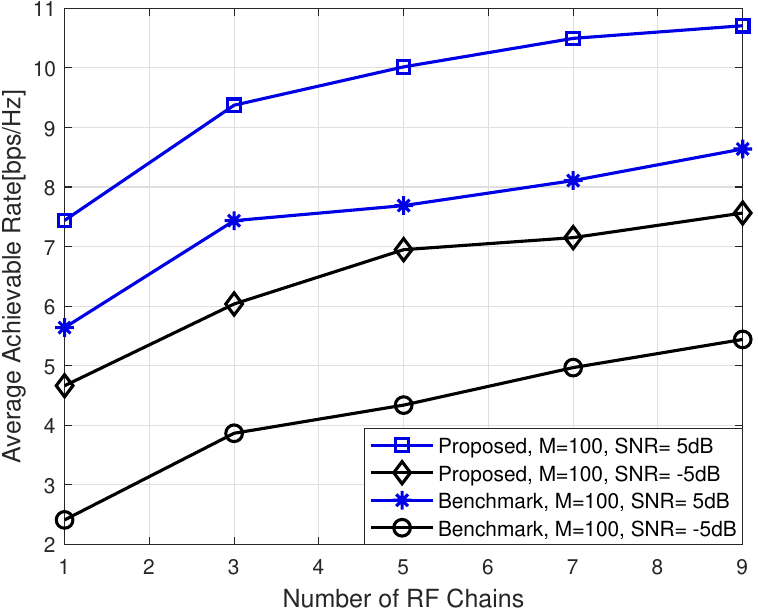}
    \caption{Average rate as a function for the number of RF chains with $M=100$.}
    \label{rate_rf_chains_M100}
\end{figure}
Figure \ref{crb_rf_chains} presents the average CRB as a function of the number of RF chains for both the proposed and benchmark schemes, with an RHS of size \(M = 36\), with two different levels of SNR.  The proposed scheme is shown to consistently achieve a lower CRB compared to the benchmark across all RF chains and at both SNR levels, demonstrating its superior performance in parameter estimation. At an SNR of $5$ dB, the CRB for both the proposed and benchmark schemes decreases significantly as the number of RF chains increases, with the proposed scheme showing a more significant reduction. At a lower SNR of $-5$ dB, the CRB is generally higher for both schemes due to the increased impact of noise, but the proposed scheme still maintains a clear advantage over the benchmark, particularly as the number of RF chains increases.

Figure \ref{rate_rf_chains_M100} shows the average achievable rate as a function of the number of RF chains, with an RHS of size \(M = 100\).  We can clearly see that even in this case the proposed scheme consistently outperforms the benchmark across all numbers of RF chains and at both SNR levels, demonstrating its superior spectral efficiency. Compared to the results shown in Figure \ref{rate_rf_chains}, we can see that an increase in the achievable rate is significant with any number of RF chains. By analyzing carefully, we can see that at an SNR of $5$ dB, the achievable rate for both the proposed and benchmark schemes increases significantly as the number of RF chains grows, with the proposed scheme showing a steeper improvement compared to the previous case of $M=36$. This clearly highlights its efficiency in leveraging additional RF chains combined with improved analog holographic beamforming due to an increased size of the RHS \(M = 100\).  At a lower SNR of $-5$ dB, the achievable rate is generally lower for both schemes, but the proposed scheme still maintains a clear advantage over the benchmark, particularly as the number of RF chains increases.

\begin{figure}
    \centering
    \includegraphics[width=0.7\linewidth]{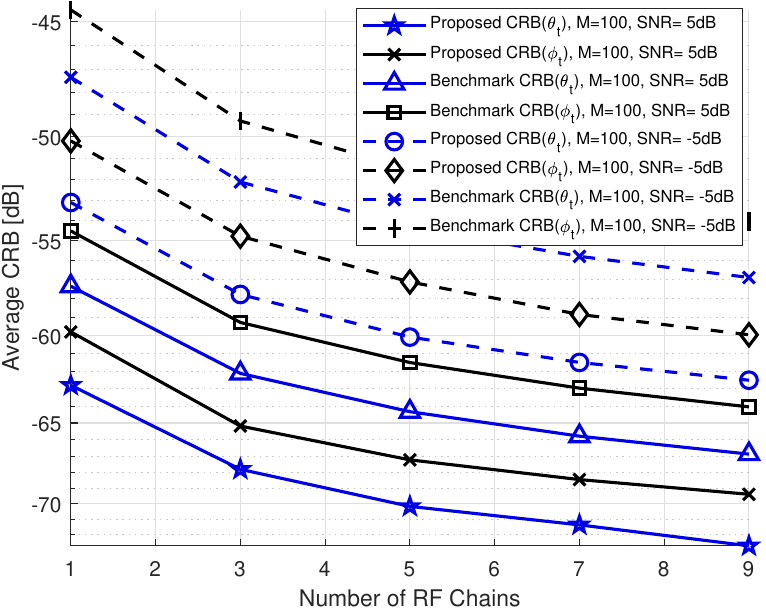}
    \caption{Average CRB as a function for the number of RF chains with $M=100$}
    \label{crb_rf_chains_M100}
\end{figure}

Finally, Figure \ref{crb_rf_chains_M100} presents the average CRB as a function of the number of RF chains for both the proposed and benchmark schemes, with an RHS of size \(M = 100\).  Compared to the previous case reported in Figure \ref{crb_rf_chains}, we can see that the proposed scheme consistently achieves a significantly lower CRB compared to the benchmark across all numbers of RF chains and at both SNR levels, demonstrating its superior performance with an increased RHS size. Both at the SNR of $-5$ and $5$ dB, the CRB for both the proposed and benchmark schemes decreases substantially as the number of RF chains increases, with the proposed scheme showing a more significant reduction, highlighting its efficiency to improve estimation accuracy, compared to the benchmark scheme.

From the results presented above, it is evident that the proposed holographic JCAS framework achieves significant performance both in terms of communication rate and sensing accuracy across all SNR levels. A key insight is the critical role of the RF chain count in system performance - doubling the number of RF chains from $3$ to $6$ substantially enhanced the performance, particularly for larger RHS  (\(M = 100\)). Even with smaller arrays (\(M = 36\)), increasing the number of RF chains leads to noticeable improvements, underscoring their importance in optimizing system functionality. These findings highlight the potential of integrating large RHS with a well-designed number of RF chains to maximize spatial resolution and degrees of freedom, paving the way for high-precision sensing and high-data-rate communication.  

%From the results reported above, several key findings and conclusions can be drawn regarding the performance of holographic JCAS systems. First, the proposed scheme consistently outperforms the benchmark across all configurations, demonstrating superior spectral efficiency and parameter estimation accuracy. This is evident from the higher achievable communication rates and lower CRBs values achieved by the proposed method, at all SNR levels. A critical observation is the significant impact of increasing the number of RF chains on system performance. Doubling the number of RF chains from $3$ to 6 led to substantial improvements in both communication rates and estimation accuracy, as reflected in the reduced CRB values. This improvement is particularly notable for larger RHS (\(M = 100\)), where the combination of more RF chains and a larger array size provides enhanced spatial resolution and resource utilization. Even with smaller arrays (\(M = 36\)), increasing the number of RF chains resulted in noticeable gains, underscoring the importance of RF chain count in optimizing system performance. Overall, we can conclude that integrating large RHS combined with reasonable number of RF chains provide greater degrees of freedom to enable holographic JCAS. This is crucial for joint high-precision parameter estimation and high communication data rate. The proposed scheme's ability to effectively exploit these resources further solidified its advantage over the benchmark.

\section{Conclusions} \label{section_5}

In this work, the authors considered a holographic JCAS system leveraging RHS to achieve high-resolution beamforming while simultaneously sensing the environment. A comprehensive theoretical framework was developed to characterize the estimation accuracy by deriving exact CRBs for azimuth and elevation angles, explicitly accounting for arbitrary antenna spacing in holographic transceivers. This result will allow for a rigorous analysis of 3D holographic sensing performance for the next generation JCAS systems. To optimize system performance, a novel weighted multi-objective problem formulation was proposed, aiming to jointly maximize the communication rate while minimizing the CRBs for enhanced sensing accuracy. To address this, the study introduced an algorithmic framework based on the MM principle, employing alternating optimization techniques. The framework utilized closed-form surrogate functions to majorize the original objective function, enabling tractable optimization. Extensive simulation results validated the effectiveness of the proposed holographic JCAS framework under various system configurations. From this investigation, we can conclude the larger RHS combined with a sufficient number of RF chains can lead to substantial improvements in the holographic JCAS systems.

\begin{appendices}
\section{Derivation of the Cramér-Rao Bounds} \label{CRB_derivation}
In this section, we derive the CRBs for 3D holographic JCAS, which sets the theoretical limits on the maximum achievable accuracy for elevation and azimuth angles.  
The received radar signal at the RHS is given as
\begin{equation}
    \mathbf{y}_{\text{r}} = \gamma \mathbf{a}(\theta_t, \phi_t) \mathbf{a}^H(\theta_t, \phi_t) \mathbf{W} \mathbf{v_d} + \mathbf{n}_{\text{sense}},
\end{equation}
The mean signal is given by
\begin{equation}
    \boldsymbol{\mu}(\boldsymbol{\eta}) = \gamma \mathbf{a}(\theta_t, \phi_t) \mathbf{a}^H(\theta_t, \phi_t) \mathbf{W} \mathbf{v_d},
\end{equation}
where \( \boldsymbol{\eta} = [\theta_t, \phi_t]^T \) are the parameters to be estimated. To derive the CRBs, we first defined the elements of the Fisher Information Matrix (FIM), which given the mean $\boldsymbol{\mu}$, can be obtained as
\begin{equation}
    [\mathbf{F}(\boldsymbol{\eta})]_{i,j} = \frac{2}{\sigma_{\text{r}}^2} \operatorname{Re} \left\{ \frac{\partial \boldsymbol{\mu}^H(\boldsymbol{\eta})}{\partial \eta_i} \frac{\partial \boldsymbol{\mu}(\boldsymbol{\eta})}{\partial \eta_j} \right\}.
\end{equation}

The derivative of \( \boldsymbol{\mu}(\boldsymbol{\eta}) \) can be written as
\begin{equation}
\begin{aligned}
     \frac{\partial \boldsymbol{\mu}(\boldsymbol{\eta})}{\partial \eta_i} = \gamma \frac{\partial \mathbf{a}(\theta_t, \phi_t)}{\partial \eta_i}& \mathbf{a}^H(\theta_t, \phi_t) \mathbf{W} \mathbf{v_d} \\&+ \gamma \mathbf{a}(\theta_t, \phi_t) \frac{\partial \mathbf{a}^H(\theta_t, \phi_t)}{\partial \eta_i} \mathbf{W} \mathbf{v_d}.
\end{aligned}
\end{equation}

By computing the derivative, it can be shown that the \((i,j)\)-th element of the FIM is given by

 \begin{equation}
\begin{aligned}
        [\mathbf{F}(\boldsymbol{\eta})]_{i,j} &= \frac{2 |\gamma|^2}{ \sigma_{\text{r}}^2} \operatorname{Re} \Bigg\{ 
        \Big( \Big( \frac{\partial  \mathbf{a}(\theta_t, \phi_t)}{\partial \eta_i} \mathbf{a}^H(\theta_t, \phi_t)  \\& \quad \quad \quad
        + \mathbf{a}(\theta_t, \phi_t) \frac{\partial  \mathbf{a}^H(\theta_t, \phi_t)}{\partial \eta_i}  \Big) \mathbf{W} \mathbf{v_d} \Big)^H \\ 
        & \quad \times \Big( \Big( \frac{\partial \mathbf{a}(\theta_t, \phi_t)}{\partial \eta_j} \mathbf{a}^H(\theta_t, \phi_t)  
        \\& \quad \quad \quad + \mathbf{a}(\theta_t, \phi_t) \frac{\partial  \mathbf{a}^H(\theta_t, \phi_t)}{\partial \eta_j} \Big) \mathbf{W} \mathbf{v_d} \Big) 
        \Bigg\}.
\end{aligned}
\end{equation}

The CRB for \( [\theta_t,\phi_t] \) are the diagonal elements of \( \mathbf{F}^{-1}(\boldsymbol{\eta}) \):
\begin{equation}
    \text{Var}(\hat{\eta}_i) \geq [\mathbf{F}^{-1}(\boldsymbol{\eta})]_{i,i}.
\end{equation}

Given the definition of the CRBs and its relationship with the FIM, we can state that the variances of the estimators \( \hat{\theta}_t \) and \( \hat{\phi}_t \) are lower-bounded as
\begin{equation}
    \text{Var}(\hat{\theta}_t) \geq [\mathbf{F}^{-1}(\boldsymbol{\eta})]_{1,1}, \quad \text{Var}(\hat{\phi}_t) \geq [\mathbf{F}^{-1}(\boldsymbol{\eta})]_{2,2}.
\end{equation}

%\subsection*{Fisher Information Matrix Element \( [\mathbf{F}(\boldsymbol{\eta})]_{1,1} \)}
In this part, we focus on the derivation of the CRB for angle $\theta_t$.
Let $\delta{\bmA_{\theta_t}}  = \frac{\partial \mathbf{a}(\theta_t, \phi_t)}{\partial \theta_t}$, and to calculate \( [\mathbf{F}(\boldsymbol{\eta})]_{1,1} \), we have \( \eta_1 = \theta_t \).We now consider differentiating \(\mathbf{a}(\theta_t, \phi_t)\) with respect to \(\theta_t\). The derivative of the steering vector is the sum of the derivatives of \(\mathbf{a}_x(\theta_t, \phi_t)\) and \(\mathbf{a}_y(\theta_t, \phi_t)\), weighted by the Kronecker product.

Let $ c_x = j k_f d_x \sin \theta_t \cos \phi_t$, then the derivative of \(\mathbf{a}_x(\theta_t, \phi_t)\) with respect to \(\theta_t\) is
\[
\frac{\partial \mathbf{a}_x(\theta_t, \phi_t)}{\partial \theta_t} = j k_f d_x \cos \theta_t \cos \phi_t \begin{bmatrix}
0 \\
  e^{c_x} \\
 2   e^{ 2 c_x} \\
\vdots \\
 \sqrt{M}-1  e^{(\sqrt{M}-1) c_x}
\end{bmatrix}.
\]
 
Let $ c_y(\theta_t,\phi_t) = j k_f d_y \sin \theta_t \sin \phi_t$,
and the derivative of \(\mathbf{a}_y(\theta_t, \phi_t)\) with respect to \(\theta_t\) is given as
\[
\frac{\partial \mathbf{a}_y(\theta_t, \phi_t)}{\partial \theta_t} = j k_f d_y \cos \theta_t \sin \phi_t \begin{bmatrix}
0 \\
  e^{c_y} \\
 2   e^{2 c_y} \\
\vdots \\
 \sqrt{M}-1  e^{(\sqrt{M}-1)c_y}
\end{bmatrix}.
\]

Using the product rule for differentiation, the derivative of \(\mathbf{a}(\theta_t, \phi_t)\) with respect to \(\theta_t\) is:
\begin{equation}
    \begin{aligned}
        \delta{\bmA_{\theta_t}} = & \frac{\partial \mathbf{a}_x(\theta_t, \phi_t)}{\partial \theta_t} \otimes \mathbf{a}_y(\theta_t, \phi_t) + \mathbf{a}_x(\theta_t, \phi_t) \otimes \frac{\partial \mathbf{a}_y(\theta_t, \phi_t)}{\partial \theta_t}.
    \end{aligned}
\end{equation}

 Let $\bmA_{\theta_t}$ be defined as

\begin{equation}
    \bmA_{\theta_t} = \delta{\bmA_{\theta_t}} \mathbf{a}(\theta_t, \phi_t)^H +  \mathbf{a}(\theta_t, \phi_t) \delta{\bmA_{\theta_t}}^H 
\end{equation}
Substituting into the FIM, we get the following final expression for the first element on the diagonal of the FIM
\begin{equation}
\begin{aligned}
    [\mathbf{F}(\boldsymbol{\eta})]_{1,1} &= \frac{2 |\gamma|^2}{\sigma_{\text{r}}^2} \Big( \mathbf{v_d}^H \mathbf{W}^H \bmA_{\theta_t}^H \bmA_{\theta_t} \mathbf{W} \mathbf{v_d}\Big). \\
\end{aligned}
\end{equation}
and by using the relationship between the FIM and the CRB for the azimuth angle \( \theta_t \), we can conclude that 
\begin{equation}
   CRB(\theta_1) = [\mathbf{F}^{-1}(\boldsymbol{\eta})]_{1,1} = \frac{\sigma_{\text{r}}^2}{2 |\gamma|^2} \Big( \mathbf{v_d}^H \mathbf{W}^H \bmA_{\theta_t}^H \bmA_{\theta_t} \mathbf{W} \mathbf{v_d}\Big)^{-1}.
\end{equation}

%\subsection*{Fisher Information Matrix Element \( [\mathbf{F}^{-1}(\boldsymbol{\eta})]_{2,2} \)}

The CRB for the elevation angle \( \phi_t \) is given by the second diagonal element of the inverse FIM. For \( [\mathbf{F}(\boldsymbol{\eta})]_{2,2} \), we have \( \eta_2 = \phi_t \). Let $\delta \bmA_{\phi_t} = \partial \mathbf{a}(\theta_t, \phi_t)/\partial \phi_t$, and we need to compute 
the following matrix
\begin{equation}
    \bmA_{\phi_t} = \delta \bmA_{\phi_t}  \mathbf{a}(\theta_t, \phi_t)^H +  \mathbf{a}(\theta_t, \phi_t) \delta \bmA_{\phi_t}^H 
\end{equation}

Using the product rule for differentiation, the derivative of \(\mathbf{a}(\theta_t, \phi_t)\) with respect to \(\phi_t\) is:
\begin{equation}
    \begin{aligned}
        \delta{\bmA_{\phi_t}} = & \frac{\partial \mathbf{a}_x(\theta_t, \phi_t)}{\partial \phi_t} \otimes \mathbf{a}_y(\theta_t, \phi_t) + \mathbf{a}_x(\theta_t, \phi_t) \otimes \frac{\partial \mathbf{a}_y(\theta_t, \phi_t)}{\partial \phi_t}.
    \end{aligned}
\end{equation}
They can be computed as

\begin{equation}
    \frac{\partial \mathbf{a}_x(\theta_t, \phi_t)}{\partial \phi_t} = - j k_f d_x \sin\theta_t \sin\phi_t
    \begin{bmatrix}
        0 \\
        e^{c_x} \\
        2 e^{  2 c_x} \\
        \vdots \\
        (\sqrt{M}-1) e^{ (\sqrt{M}-1)c_x}
    \end{bmatrix},
\end{equation}

\begin{equation}
    \frac{\partial \mathbf{a}_y(\theta_t, \phi_t)}{\partial \phi_t} = j k_f d_y \sin\theta_t \cos\phi_t
    \begin{bmatrix}
        0 \\
        e^{c_y} \\
        2 e^{  2 c_y} \\
        \vdots \\
        (\sqrt{M}-1) e^{ (\sqrt{M}-1)c_y}
    \end{bmatrix},
\end{equation}

Substituting the partial derivatives into the FIM we get
\begin{equation}
\begin{aligned}
    [\mathbf{F}(\boldsymbol{\eta})]_{2,2} =  \frac{2 |\gamma|^2}{\sigma_{\text{r}}^2} \Big( \mathbf{v_d}^H \mathbf{W}^H \bmA_{\theta_t}^H \bmA_{\theta_t} \mathbf{W} \mathbf{v_d}\Big).
\end{aligned}
\end{equation}

and CRB for the elevation angle \( \phi_t \) can be obtained as
\begin{equation}
    CRB(\phi_t)[\mathbf{F}^{-1}(\boldsymbol{\eta})]_{2,2} = \frac{\sigma_{\text{r}}^2}{2 |\gamma|^2} \Big( \mathbf{v_d}^H \mathbf{W}^H \bmA_{\phi_t} \bmA_{\phi_t}^H \mathbf{W} \mathbf{v_d}\Big)^{-1}.
\end{equation}

\section{Proof of Theorem 2} \label{proof_Thm2}
\begin{proof}
The proof relies on the first-order Taylor expansion of the function \(f(x) = x^{-1}\) around a point \(x^{(t)} > 0\). The Taylor expansion provides a linear approximation of \(f(x)\) near \(x^{(t)}\), given by:
\begin{align}
    f(x) \approx f(x^{(t)}) + f'(x^{(t)})(x - x^{(t)}),
\end{align}
where \(f'(x) = -x^{-2}\) is the derivative of \(f(x)\). 

Substituting \(x = \mathbf{v_d}^H \mathbf{B} \mathbf{v_d}\) and \(x^{(t)} = \mathbf{v_d}^{(t)H} \mathbf{B} \mathbf{v_d}^{(t)}\), we obtain:
\begin{align}
    \frac{1}{\mathbf{v_d}^H \mathbf{B} \mathbf{v_d}} \approx & \frac{1}{\mathbf{v_d}^{(t)H} \mathbf{B} \mathbf{v_d}^{(t)}} - \frac{1}{\left(\mathbf{v_d}^{(t)H} \mathbf{B} \mathbf{v_d}^{(t)}\right)^2} \cdot \\
    & \quad \quad \quad \quad  \Big(\mathbf{v_d}^H \mathbf{B} \mathbf{v_d}  - \mathbf{v_d}^{(t)H} \mathbf{B} \mathbf{v_d}^{(t)} \Big).
\end{align}

This approximation can be rewritten as:
\begin{align}
    \frac{1}{\mathbf{v_d}^H \mathbf{B} \mathbf{v_d}} \approx \frac{2}{\mathbf{v_d}^{(t)H} \mathbf{B} \mathbf{v_d}^{(t)}} - \frac{\mathbf{v_d}^H \mathbf{B} \mathbf{v_d}}{\left(\mathbf{v_d}^{(t)H} \mathbf{B} \mathbf{v_d}^{(t)}\right)^2}.
\end{align}

To establish the inequality, we note that the function \(f(x) = x^{-1}\) is convex for \(x > 0\). By the properties of convex functions, the first-order Taylor expansion provides a global upper bound. Therefore, the approximation can be replaced with an inequality:
\begin{align}
    \frac{1}{\mathbf{v_d}^H \mathbf{B} \mathbf{v_d}} \leq \frac{2}{\mathbf{v_d}^{(t)H} \mathbf{B} \mathbf{v_d}^{(t)}} - \frac{\mathbf{v_d}^H \mathbf{B} \mathbf{v_d}}{\left(\mathbf{v_d}^{(t)H} \mathbf{B} \mathbf{v_d}^{(t)}\right)^2}.
\end{align}

This completes the proof, as the derived inequality provides a valid surrogate function for the inverse quadratic term \((\mathbf{v_d}^H \mathbf{B} \mathbf{v_d})^{-1}\).
\end{proof}
\end{appendices}

{\footnotesize
\bibliographystyle{IEEEtran}
\def\baselinestretch{0.90}
\bibliography{main}}
\end{document}